\documentclass[aps,pra,twocolumn]{revtex4-2}

\usepackage[T1]{fontenc} 
\usepackage{babel}

\usepackage{microtype}

\usepackage{xcolor}   
\definecolor{urlcol}{RGB}{0, 102, 204}

\usepackage{hyperref}
\hypersetup{
  colorlinks = true,
  urlcolor = urlcol,
  linkcolor = black,
  citecolor = blue,
  breaklinks,
}

\usepackage{graphicx}   
\usepackage{float}    

\usepackage{etoolbox} 

\AfterEndEnvironment{figure}{\noindent\ignorespaces}
\AfterEndEnvironment{figure*}{\noindent\ignorespaces}

\usepackage{mathtools} 
\usepackage{amssymb} 
\usepackage{siunitx} 
\usepackage{dsfont} 
\usepackage{icomma} 
\usepackage{bm} 

\usepackage{IEEEtrantools}

\makeatletter
\let\if@IEEElastlinewassubequation\iffalse%
\@ifundefined{c@IEEEsubequation}{%
  \def\@IEEEeqnarrayXCR[#1]{%
    \if@eqnsw%
      \if@IEEEissubequation%
      \else%
        \refstepcounter{equation}%
        \addtocounter{equation}{-1}%
      \fi%
    \fi%
    \cref@orig@IEEEeqnarrayXCR[#1]}%
}{
  \def\@IEEEeqnarrayXCR[#1]{%
    \if@eqnsw%
      \ifnum\c@IEEEsubequation>0\relax%
      \else%
        \refstepcounter{equation}%
        \addtocounter{equation}{-1}%
      \fi%
    \fi%
    \cref@orig@IEEEeqnarrayXCR[#1]}%
}%
\makeatother

\makeatletter
\def\IEEElabelanchoreqn#1{\bgroup
\def\@currentlabel{\p@equation\theequation}\relax
\def\@currentHref{\@IEEEtheHrefequation}\label{#1}\relax
\Hy@raisedlink{\hyper@anchorstart{\@currentHref}}\relax
\Hy@raisedlink{\hyper@anchorend}\egroup}

\makeatother

\usepackage{amsthm}  

 \usepackage[nolist,nolinebreak]{acronym}  

\usepackage[english,capitalize]{cleveref}

\newtheorem{theorem}{Theorem}[section]
\newtheorem{lemma}[theorem]{Lemma}
\newtheorem{proposition}[theorem]{Proposition}

\NewDocumentCommand\captionShortcut{ommo}{
	\IfValueTF{#1}{\caption[#1]{#2}}{\caption{#2}}
	\label{#3:\IfValueTF{#4}{#4}{\IfValueTF{#1}{#1}{#2}}}
}

\NewDocumentCommand\subcaptionShortcut{omm}{
	\subcaption{\IfValueTF{#1}{#1}{}\hfill\hphantom{}}
	\label{#2:#3}
}

\usepackage{xparse}

\DeclarePairedDelimiter\Parentheses{\lparen}{\rparen}
\DeclarePairedDelimiter\Brackets{\lbrack}{\rbrack}
\DeclarePairedDelimiter\Braces{\lbrace}{\rbrace}
\DeclarePairedDelimiter\AbsouluteValue{\lvert}{\rvert}
\DeclarePairedDelimiter\Norm{\lVert}{\rVert}
\DeclarePairedDelimiter\Bra{\langle}{\rvert}
\DeclarePairedDelimiter\Ket{\lvert}{\rangle}
\DeclarePairedDelimiter\BraKet{\langle}{\rangle}
\DeclarePairedDelimiterX\BraKetTwo[2]{\langle}{\rangle}{%
  #1\delimsize\vert\mathopen{}#2}
\DeclarePairedDelimiterX\BraKetThree[3]{\langle}{\rangle}{%
  #1\delimsize\vert\mathopen{}#2%
  \delimsize\vert\mathopen{}#3%
}
\DeclarePairedDelimiterX\BraKetTwoStar[2]{\langle}{\rangle}{%
  #1\nonscript\,\delimsize\vert\nonscript\,\mathopen{}#2}
\DeclarePairedDelimiterX\BraKetThreeStar[3]{\langle}{\rangle}{%
  #1\nonscript\,\delimsize\vert\nonscript\,\mathopen{}#2\nonscript\,%
  \delimsize\vert\nonscript\,\mathopen{}#3%
}

\DeclarePairedDelimiter\BraceDot{\lbrace}{.}
\DeclarePairedDelimiter\DotBrace{.}{\rbrace}

\DeclarePairedDelimiterX\SetVerticalLine[2]{\lbrace}{\rbrace}{%
  #1\nonscript\:\delimsize\vert\allowbreak\nonscript\:\mathopen{}#2}
\DeclarePairedDelimiterX\SetColon[2]{\lbrace}{\rbrace}{%
  #1\nonscript\::\allowbreak\nonscript\:\mathopen{}#2}

\DeclareMathOperator\Span{span} 
\DeclareMathOperator\emptyOperator{}

\NewDocumentCommand\diff{om}{%
\mathop{}\!\mathrm{d}\IfValueTF{#1}{^{#1}\!}{}#2\emptyOperator\!}

\newcommand\pa[1]{\Parentheses*{#1}}
\newcommand\bk[1]{\Brackets*{#1}}
\newcommand\bc[1]{\Braces*{#1}}
\newcommand\abs[1]{\AbsouluteValue*{#1}}
\newcommand\norm[1]{\Norm*{#1}}

\newcommand\bcdot[1]{\BraceDot*{#1}}
\newcommand\dotbc[1]{\DotBrace*{#1}}

\providecommand\bra{}
\providecommand\ket{}
\providecommand\braket{}
\RenewDocumentCommand\bra{smd<>}{
  \IfBooleanTF{#1}{\Bra{#2}\IfValueTF{#3}{_{#3}}{}}{\Bra*{#2}\IfValueTF{#3}{_{#3}}{}}
}
\RenewDocumentCommand\ket{smd<>}{
  \IfBooleanTF{#1}{\Ket{#2}\IfValueTF{#3}{_{#3}}{}}{\Ket*{#2}\IfValueTF{#3}{_{#3}}{}}
}
\RenewDocumentCommand\braket{smood<>}{
  \IfBooleanTF{#1}
  {
    \IfValueTF{#3}{\IfValueTF{#4}{\BraKetThree{#2}{#3}{#4}}{\BraKetTwo{#2}{#3}}}%
    {\BraKet{#2}}\IfValueTF{#5}{_\IfValueTF{#5}{#5}{#3}}{}%
  }
  {
    \IfValueTF{#3}{\IfValueTF{#4}{\BraKetThreeStar*{#2}{#3}{#4}}{\BraKetTwoStar*{#2}{#3}}}%
    {\BraKet*{#2}}\IfValueTF{#5}{_\IfValueTF{#5}{#5}{#3}}{}%
  }
}
\NewDocumentCommand\ketbra{smd<>od<>d>>}{
  \IfBooleanTF{#1}
  {
    \Ket{#2}\IfValueTF{#6}{_{\IfValueTF{#3}{#3}{#6}}}{\IfValueTF{#3}{_{#3}}{}}\Bra*{\IfValueTF{#4}{#4}{#2}}%
    \IfValueTF{#6}{_{\IfValueTF{#5}{#5}{#6}}}{\IfValueTF{#5}{_{#5}}{}}%
  }
  {
    \Ket*{#2}\IfValueTF{#6}{_{\IfValueTF{#3}{#3}{#6}}}{\IfValueTF{#3}{_{#3}}{}}\Bra*{\IfValueTF{#4}{#4}{#2}}%
    \IfValueTF{#6}{_{\IfValueTF{#5}{#5}{#6}}}{\IfValueTF{#5}{_{#5}}{}}%
  }
}

\newcommand\matp[1]{\begin{pmatrix}#1\end{pmatrix}}
\newcommand\matb[1]{\begin{bmatrix}#1\end{vmatrix}}

\newcommand*\ee{\mathrm{e}}

\newcommand*\up{\uparrow}
\newcommand*\down{\downarrow}

\newcommand*\R{\mathbb{R}}

\NewDocumentCommand\cnot{o}{\mathrm{CNOT}\IfValueTF{#1}{_{#1}}{}}
\NewDocumentCommand\Swap{o}{\mathrm{SWAP}\IfValueTF{#1}{_{#1}}{}}

\NewDocumentCommand\order{o}{\mathcal{O}\IfValueTF{#1}{\pa{#1}}{}}

\NewDocumentCommand\Time{o}{\bm{\mathrm{TIME}}\IfValueTF{#1}{\pa{#1}}{}}
\NewDocumentCommand\Space{o}{\bm{\mathrm{SPACE}}\IfValueTF{#1}{\pa{#1}}{}}

\newcommand*\eqspace{\mathrel{\phantom{=}}}

\AtBeginDocument{%
\renewcommand\url[1]{\href{https://\detokenize{#1}}{\detokenize{#1}}}%
\renewcommand\doi[1]{\href{https://doi.org/\detokenize{#1}}{\detokenize{#1}}}%
}

\makeatletter
\def\bstctlcite{\@ifnextchar[{\@bstctlcite}{\@bstctlcite[@auxout]}}
\def\@bstctlcite[#1]#2{\@bsphack
  \@for\@citeb:=#2\do{%
    \edef\@citeb{\expandafter\@firstofone\@citeb}%
    \if@filesw\immediate\write\csname #1\endcsname{\string\citation{\@citeb}}\fi}%
  \@esphack}
\makeatother

\newcommand{\lcref}[1]{\lcnamecref{#1}~\labelcref{#1}}
\makeatletter
\def\lcfirstnamecrefs#1,#2\@nil{\lcnamecrefs{#1}}
\newcommand{\lcfirstnamecref}[1]{\lcfirstnamecrefs #1,\@nil}
\newcommand{\lcrefs}[1]{\lcfirstnamecref{#1}~\labelcref{#1}}
\makeatother
\let\lcref\cref
\let\lcrefs\cref

\begin{document}

\bstctlcite{config}

\makeatletter
\let\stashsmartcomma\sm@rtcomma
\let\sm@rtcomma,
\makeatletter

\title{Digital quantum simulation of the BCS model with a central-spin-like quantum processor}
\author{Jannis Ruh}
\email[]{jannis.ruh@uni-konstanz.de}
\author{Regina Finsterhoelzl}
\email[]{regina.finsterhoelzl@uni-konstanz.de}
\author{Guido Burkard}
\email[]{guido.burkard@uni-konstanz.de}
\affiliation{Department of Physics, University of Konstanz, D-78457 Konstanz, Germany}


\begin{abstract}
The simulation of quantum systems is one of the most promising applications of quantum computers. In this paper we present a quantum algorithm to perform digital quantum simulations of the (reduced) \ac{bcs} model on a quantum register with a star shaped connectivity map, as it is, e.g., featured by color centers in diamond. We show how to effectively translate the problem onto the quantum hardware and implement the algorithm using only the native interactions between the qubits. Furthermore we discuss the complexity of the circuit. We use the algorithm to simulate the dynamics of the \acs{bcs} model by subjecting its mean-field ground state to a time-dependent perturbation. The quantum simulation algorithm is studied using a classical simulation.
\end{abstract}

\maketitle

\makeatletter
\let\sm@rtcomma\stashsmartcomma
\makeatother

\section{Introduction}

The current state of quantum computing hardware platforms has been termed the era of \ac{nisq} computers \cite{preskill_nisq}, thereby referring to their limitations due to gate errors and decoherence effects. However, recent rapid developments may soon lead to the demonstration of advantages of useful quantum or hybrid algorithms over pure classical algorithms \cite{arute_quantum_supremacy,mohseni_commercialization}.
Quantum algorithms \cite{montanaro_overview} have a broad area of applications, from the generalized
Shor algorithm for the solution of the hidden subgroup problem \cite{shor_algorithm, lomonaco_hidden_subgroup} and quantum approximate optimization \cite{farhi_qaoa, farhi_qaoa_supremacy} to the simulation of real quantum systems \cite{lloyd_simulating_universality, abrams_fermi_simulation, aspuru_molecular_energies, haah_local_hamiltonian_simulation, hastings_improve_jordan_wigner, ortiz_fermionic_simulations, wecker_fermionic_simulations, wecker_practical_variational_algorithm, wecker_simulation_gate_count}. The goal of these algorithms is to solve problems whose high computational cost makes them hard or even impossible to solve with classical hardware. Particularly the simulation of quantum systems is among these problems due to the exponentially large dimension of the state space \cite{feynman_physics_with_computers}. Since the currently available quantum hardware platforms are limited, it is important to develop implementations of quantum algorithms that make optimal use of the available hardware. To achieve this, the algorithms can be aligned with the structure of the quantum processor, i.e., with the coupling map which describes the possible connections between the qubits. Because of the limited number of available qubits and the need to protect them against decoherence and error-prone gates, it is desirable to minimize the number of operations that are required to translate the quantum algorithm to the hardware \cite{gottesman_error_correction, aharonov_fault_tolerant, bravyi_lattice_code, dennis_topological_memory, fowler_surface_code, shor_error_code}.

As small quantum systems only require a limited number of logical qubits, their simulation on the current \ac{nisq} devices has already been demonstrated for very small systems \cite{omalley_simulation_experiment}. There already exist many quantum algorithms to perform such tasks \cite{lloyd_simulating_universality, abrams_fermi_simulation, aspuru_molecular_energies, haah_local_hamiltonian_simulation, hastings_improve_jordan_wigner, ortiz_fermionic_simulations, wecker_fermionic_simulations, wecker_practical_variational_algorithm, wecker_simulation_gate_count}. For instance,  \cite{wecker_fermionic_simulations} presents an algorithm that may be used to analyze the ground state and phase diagram of the Hubbard model. However, most of the algorithms do not consider any restrictions given by the structure of the quantum hardware. This may cause the transpilation to be costly in terms of additionally needed gates.

In this paper, we present an implementation of the quantum simulation of the \acf{bcs}
model for superconductivity. Our implementation is an example of a Hamiltonian simulation, where the quantum time evolution of a system is simulated. We will restrict the physical system to the space of Cooper pairs, which enables us to map the system efficiently onto a spin system with $\order[1]$. This improves the performance of the algorithm, however, it also implies that the presented quantum circuit effectively simulates a spin model and not a fermionic model. To simulate the whole fermionic system one has to use a fermionic mapping such as the Jordan-Wigner mapping \cite{tranter_bravyi_kitaev_transformation}. While there exist analytical solutions for a time independent system \cite{yuzbashyan_dynamics}, our numerical quantum algorithm is applicable to the simulation of time-dependent problems and can be extended, by using trotterization techniques, to include perturbation terms. The error of the algorithm is only of numerical nature, which can, theoretically, be reduced to be arbitrarily small. This is in contrast to analytical approximations. We restrict ourselves to the state space of paired electrons, the Cooper pairs. The algorithm is tailored to a quantum computer with a coupling map based on a \ac{css}. Such a quantum computer can, for instance, be realized with a spin-qubit register consisting of a nitrogen-vacancy defect in diamond \cite{childress_dynamics_coupled_electron_nuclear_spins_in_diamond,neumann_multipartite_entanglement_in_diamond,jiang_readout_in_diamond,dutt_register_in_diamond,steiner_optical_readout_in_diamond,smeltzer_control_of_spins_in_diamond,waldherr_error_correction_in_diamond,dolde_entanglement_in_diamond,kalb_memory_in_diamond,pezzagna_summary_computing_in_diamond,wrachtrup_computing_in_diamond,doherty_summary_nv_in_diamond,robledo_readout_in_diamond,finsterhoelzl2022,Vorobyov2022}. In addition to the simulation of the BCS model, the proposed algorithm offers an efficient implementation for multi-qubit gates that are double products of two-qubit gates on a \ac{css}-like quantum register.
\begin{figure}[hbt!]
\center
\includegraphics{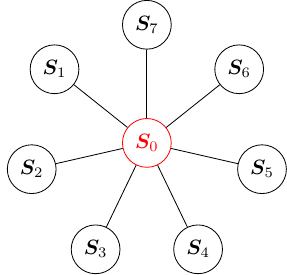}
\caption{The coupling map of a quantum register based on a \acl{css} with eight qubits. The central qubit ($\bm{S}_{0}$) is connected to all other qubits ($\bm{S}_{1}, \ldots, \bm{S}_{7}$), while the other qubits are not connected with each other.}
\label{fig:nv sketch}
\end{figure}

The paper has the following structure: In \lcref{sec:The models}, we introduce the physical model of a
\ac{bcs} superconductor, the simulated quantum system. Next we establish the connectivity map of quantum computer based on a \ac{css}. In the \lcrefs{sec:Mapping on a css,sec:Mapping on a quantum computer}, we show how to decompose the \ac{bcs} Hamiltonian into \ac{css}-like Hamiltonians and perform the mapping of the physical problem onto a quantum computer. \Cref{sec:Simulation} describes the quantum algorithm and in \lcref{sec:Results} we present our numerical results, where we simulate the time evolution of the mean-field ground state using a simulated quantum computer. We perform a quench, i.e., an abrupt parameter change in time, as a possible application of the algorithm, and discuss how the algorithm can be improved.

\section{The model}  \label{sec:The models}

The \ac{bcs} theory was introduced by J. Bardeen, L. N. Cooper, and J. R. Schrieffer to describe the phenomenon of
superconductivity through the pairing of electrons in a metal \cite{bcs_theory_of_superconductivity}. In the case of discrete states (e.g., in metallic grains) where the level spacing is of the order of the superconducting energy gap, a reduced BCS model can be used \cite{delft_spectroscopy_in_ultrasmall_grains,delft_superconductivity_in_ultrasmall_grains_intro_richardsons_solution,dukelsky_solvable_richardson_gaudin_models}. The Hamiltonian can then be written in the form
\cite{yuzbashyan_dynamics,delft_superconductivity_in_ultrasmall_grains_intro_richardsons_solution}
\begin{equation}  \label{eq:base hamiltonian}
  H_{\mathrm{BCS}} = \sum_{j=0}^{n-1} \sum_{\sigma = \up, \down} \epsilon_{j}
  c^{\dagger}_{j\sigma} c_{j\sigma} - g \sum_{j,k=0}^{n-1} c^{\dagger}_{j\up}
  c^{\dagger}_{j\down} c_{k\down} c_{k\up} \, ,
\end{equation}
where the first term corresponds to the single-particle Hamiltonian with fermionic operators
$c^{\dagger}_{j\sigma}$ and $c_{j\sigma}$ describing the creation and annihilation of
electrons in orbital $j$ with energy $\epsilon_{j}$ and spin $\sigma$, respectively. The second
term describes an effective pairwise interaction between the electrons where we assume a constant and energy-level independent coupling strength $g$. This coupling arises as the result of a perturbative description of the interaction between the electrons and phonons
\cite{anderson_dirty_bcs}. The pairing takes place between states of equal energy but antiparallel spins, i.e., between $\ket{j, \up}$ and $\ket{j,
\down}$, which occurs in a system with time reversal
invariance where the single-particle energy levels $\epsilon_{j}$ are only degenerate
with respect to the spin \cite{yuzbashyan_dynamics} or if the basis states are real wave functions. The number $n$ of energy orbitals is assumed to be finite, for example, as in models that describe superconductivity in ultrasmall metallic grains with an energy cutoff \cite{braun_fixed,delft_superconductivity_in_ultrasmall_grains_intro_richardsons_solution}. Under the assumption of constant parameters $g$ and $\epsilon_{j}$, the Hamiltonian in \cref{eq:base hamiltonian} is a well-described integrable model \cite{cambiaggio_integrability_pairing}.
The \ac{bcs} model \cref{eq:base hamiltonian} can be used to calculate the superconducting energy gap $\Delta$, see 
\cref{sec:Mean-field ground state}.
Analytical mean-field solutions, derived with algebraic separation of variables methods \cite{sklyanin_separation_of_variables, kutznetsov_separation_of_variables, sklyanin_gaudin_models}, exist for its dynamics \cite{yuzbashyan_dynamics}. In contrast to that, the quantum algorithm to be proposed here is able to simulate the \ac{bcs} system with time-dependent parameters and can be easily extended to include perturbation terms.

A \acl{css} can be described by a Hamiltonian of the following form \cite{froehling_central_spin}
\begin{equation}
  H_{\text{CSS}} = \sum_{j = 1}^{n-1} J_{j} \bm{S}_{0} \cdot \bm{S}_{j} + B \sum_{j =
  0}^{n-1} \mu_{j} S^{z}_{j} \,,  \label{eq:central spin}
\end{equation}
where $\bm{S}_{j} = \pa{S^{x}_{j}, S^{y}_{j}, S^{z}_{j}}^{T}$ are spin-$\frac{1}{2}$ operators, defined via the matrix representation $\pa{S^{x}_{j}, S^{y}_{j}, S^{z}_{j}}^{T} = \frac{\hbar}{2}\pa{\sigma^{x}, \sigma^{y}, \sigma^{z}}^{T}$, with the standard Pauli matrices
\begin{equation}
  \sigma^{x} = \matp{0&1\\1&0} \,,\quad \sigma^{y} = \matp{0&-i\\i&0} \,,\quad \sigma^{z} = \matp{1&0\\0&-1} \,. \label{eq:Pauli matrices}
\end{equation}
The first
term on the right-hand side of \cref{eq:central spin} represents the three-dimensional Heisenberg interaction between the spins, where the sign of the
coupling constant $J_{j}$ determines whether the interaction is ferromagnetic or
antiferromagnetic. The interaction only appears between one central spin $\bm{S}_0$ with the surrounding spins. Thus, the connectivity map of a quantum register
based on a \ac{css} is star-shaped as shown in \cref{fig:nv sketch}. Native two-qubit gates only exist between the central spin and the surrounding spins. This is
in contrast to ideal quantum computers, where an all-to-all connectivity is assumed. The second term in \cref{eq:central spin} describes the
Zeeman interaction with a magnetic field of strength $B$ and
coupling constants $\mu_{j} = g_{j} \mu_{\mathrm{B}}/\hbar$
with the Lande factor $g_{j}$ and the Bohr magneton $\mu_{\mathrm{B}}$.

\section{Mapping onto a central spin system} \label{sec:Mapping on a css}

The total Hilbert space of the considered system in \cref{eq:base hamiltonian} is given
by $\mathcal{H} = \bigotimes_{j=0}^{n-1} \mathcal{H}_{j}= \Span\bk{\bigotimes_{j=0}^{n-1}\mathcal{B}_{j}}$ with $\mathcal{H}_{j} =
\Span\mathcal{B}_j$ and the bases $\mathcal{B}_{j} = \bc{\Ket{0},
  c^{\dagger}_{j\down}\Ket{0}, c^{\dagger}_{j\up}\Ket{0},
c^{\dagger}_{j\up}c^{\dagger}_{j\down}\Ket{0}}$. In equal manner, we define the space of
Cooper-pairs $\mathcal{H}_{\mathrm{C}} = \bigotimes_{j=0}^{n-1}
\mathcal{H}_{\mathrm{C},j}$ with $\mathcal{H}_{\mathrm{C}, j} =
\Span\mathcal{B}_{\mathrm{C}, j}$ and $\mathcal{B}_{\mathrm{C}, j} = \bc{\Ket{0},
c^{\dagger}_{j\down}c^{\dagger}_{j\up}\Ket{0}}$. The orthogonal complements of
$\mathcal{H}_{\mathrm{C}, j}$ and $\mathcal{H}_{\mathrm{C}}$ are defined via
$\mathcal{H}^{\perp}_{\mathrm{C}, j} = \Span \bk{\mathcal{B}_{j} \backslash
  \mathcal{B}_{\mathrm{C}, j}}$ and $\mathcal{H}^{\perp}_{\mathrm{C}} = \Span
  \bk{\pa{\bigotimes_{j=0}^{n-1}\mathcal{B}_{j}} \middle\backslash
  \pa{\bigotimes_{j=0}^{n-1}\mathcal{B}_{\mathrm{C}, j}}}$, respectively.

In order to bring the \ac{bcs} Hamlitonian in connection with the Hamiltonian of a
\acl{css}, we define the operators:
\begin{align}
  K^{z}_{j} &= \frac{1 - c^{\dagger}_{j\up}c_{j\up} -
  c^{\dagger}_{j\down}c_{j\down}}{2} \,,\\
  K^{+}_{j} &= c^{\dagger}_{j\up}c^{\dagger}_{j\down} \,,\\
  K^{-}_{j} &= c_{j\down}c_{j\up} \,.
\end{align}
$K^{z}_{j}$ denotes the (shifted, negative) number operator for the  $j\text{th}$ orbital, and $K^{+}_{j}$
($K^{-}_{j}$) creates (annihilates) a Cooper-pair in the $j\text{th}$ orbital. With
these operators one can rewrite the Hamiltonian in \cref{eq:base hamiltonian} into the following form
\begin{equation}
  H_{\text{BCS}} = -\sum_{j = 0}^{n - 1} 2 \epsilon_{j} K^{z}_{j} - g \sum_{j,k = 0}^{n-1}
  K^{+}_{j} K^{-}_{k} + \mathrm{const}\,.
\end{equation}
The operators $K_j^{\alpha}$ effectively represent spin-$\frac{1}{2}$ operators.
Let us define the remaining components $K^{x}_{j} = \frac{K^{+}_{j} + K^{-}_{j}}{2}$ and
$K^{y}_{j} = \frac{K^{-}_{j} - K^{+}_{j}}{2i}$ and set $\bm{K}_{j} = \pa{K^{x}_{j},
K^{y}_{j}, K^{z}_{j}}^{T}$. One can easily show that these operators fulfill $\bm{K}_{j}
\mathcal{H}^{\perp}_{\mathrm{C}, j} = 0$ and $\bm{K}_{j} \mathcal{H}_{\mathrm{C}, j}
\subseteq \mathcal{H}_{\mathrm{C}, j}$. Moreover if we map $\Ket{0} \to
\pa{1, 0}^{T}$ and
$\hat{c}^{\dagger}_{j\down}\hat{c}^{\dagger}_{j\up}\Ket{0} \to
\pa{0, 1}^{T}$ we find the mapping
$\bm{K}_{j}\vert_{\mathcal{H}_{\mathrm{C}, j}} \to \frac{1}{2} \bm{\sigma}$, where
$\bm{\sigma} = \pa{\sigma^{x}, \sigma^{y}, \sigma^{z}}^{T}$ represents the vector of Pauli matrices, as in \cref{eq:Pauli matrices}.

Before we map the \ac{bcs} problem on a \ac{css} based quantum computer, we
introduce the Gaudin
Hamiltonians \cite{sklyanin_gaudin_models}, a family of operators similar to a \ac{css} Hamiltonian. For simplicity, we assume non-degenerate
energy levels $\epsilon_{j}$, however the following can also be generalized via introducing summed-spin operators $\bm{\tilde{K}}_{j} = \sum_{k = 0}^{n - 1} \bm{K}_k \delta_{\epsilon_j, \epsilon_k}$ and describing the \ac{bcs} Hamiltonian with these operators. For $q \in \bc{0, \ldots,
n-1}$ we define the Gaudin Hamiltonians as
\begin{equation}
  H_{q} = 2 \sum_{\substack{j=0\\j\neq q}}^{n-1} \frac{\bm{K}_{q} \cdot
  \bm{K}_{j}}{\epsilon_{q} - \epsilon_{j}} - \gamma K^{z}_{q} \,,
\end{equation}
with the free parameter $\gamma$. These Hamiltonians can be seen as special case of the \ac{css}
Hamiltonian in \cref{eq:central
spin},  if we assume either constant $\mu_{j} = \mu$, for $j = 1, \ldots, n - 1$, with a
conserved total spin, or $\mu_{j} \ll \mu_{0}$ for $j = 1, \ldots, n - 1$. The
central spin is at index $q$. If we choose $\gamma = - \frac{2}{g}$, the Gaudin
Hamiltonians represent a set of invariants with respect to the \ac{bcs}
Hamiltonian \cite{cambiaggio_integrability_pairing}, i.e, $\bk{H_{\mathrm{BCS}}, H_{q}} = 0$ and $\bk{H_{p},
H_{q}} = 0$. Moreover we can construct the \ac{bcs} Hamiltonian with them,
\begin{equation}
\label{eq:starting Hamiltonian} H_{\mathrm{BCS}} = -g \sum_{q=0}^{n-1}\epsilon_{q} H_{q}
+ g L^{z} + g \pa{L^{z}}^{2} \,,
\end{equation} where we used the total spin operator
$\bm{L}$ that is defined as $\bm{L} = \sum_{j=0}^{n-1} \bm{K}_{j}$. Note that the Gaudin Hamiltonians also fulfill $-\gamma L^{z} = \sum_{q=0}^{n-1} H_{q}$. If the energy levels $\epsilon_{j}$ are degenerate, there is an additional term in the Hamiltonian in \cref{eq:starting Hamiltonian} (see \lcref{sec:Properties of the Gaudin Hamiltonians} for more details).\cite{frenkel_gaudin_commutator_proof, sklyanin_gaudin_models}

\section{Mapping onto a quantum computer}  \label{sec:Mapping on a quantum computer}

We restrict ourselves to the Hilbert space $\mathcal{H}_{\mathrm{C}}$ of Cooper-pairs. This enables us
to map the operators $\bm{K}_j$ with a resource overhead of order $\mathcal{O}\pa{1}$ onto a quantum computer. This stands in
contrast to other mappings, e.g., the Jordan-Wigner mapping, which maps creation and
annihilation operators to Pauli operators with an overhead of order $\mathcal{O}\pa{n}$ where $n$ is the number of
qubits, or the Bravyi-Kitaev mapping with a mapping order of $\mathcal{O}\pa{\log
n}$ \cite{tranter_bravyi_kitaev_transformation}.
As demonstrated in \lcref{sec:Mapping on a css} the space $\mathcal{H}_{\mathrm{C}}$ is
invariant under the action of the operators $\bm{K}_j$. This implies that
$\mathcal{H}_{\mathrm{C}}$ is also invariant under $H_{\mathrm{BCS}}$ (\cref{eq:base hamiltonian}), since
$H_{\mathrm{BCS}}$ can be expressed through the spin operators as in \cref{eq:starting
Hamiltonian}. Moreover, since the Hamiltonian $H_{\mathrm{BCS}}$ is hermitian, it is block diagonal with respect to $\mathcal{H}_{\mathrm{C}}$ and its complement $\mathcal{H}^{\perp}_{\mathrm{C}}$. From a physical point of view this is caused by the fact that the interaction term in \cref{eq:base hamiltonian} only rearranges the energy levels that are occupied by Cooper pairs, it does not break up or create any Cooper pairs into or out of single occupied energy levels, respectively. The block diagonal form enables us to consider the Hilbert space $\mathcal{H}_{\mathrm{C}}$ as a self-contained system.

The mapping onto qubits is done via
\begin{equation}
  \prod_{j = 0}^{n-1} \pa{K^{+}_{j}}^{\beta_{i}} \ket{0} \to \ket{\{q_{j}\}} \,,
\end{equation}
with $q_{j} = \beta_{j} \in \bc{0, 1}$ and
\begin{equation}
  \bm{K}_{j} \to \frac{1}{2} \bm{\sigma}_j \,,
\end{equation}
where $\ket{\{q_{j}\}} = \ket{q_{n-1} \ldots q_{0}}$ represents a basis state of the qubits on the quantum computer. Here we make use of the possibility to represent the operators $\bm{K}_{j}$ with the Pauli operators in
$\mathcal{H}_{\mathrm{C}}$, as described in \lcref{sec:Mapping on a css}. This mapping is similar to the proposed mapping in \cite{wu_pairing_models_simulation}, where different kinds of pairings are investigated. Note that we do not have to consider any parity signs
caused by fermionic anti-commutators, since the fermionic creation and annihilation operators always appear pairwise. This makes the proposed mapping more efficient than the mapping of single creation and
annihilation operators.

\section{Simulation} \label{sec:Simulation}

Let us first consider the case with constant parameters $\epsilon_j$ and $g$, and without perturbation terms. The time evolution operator at time $t$ of the \ac{bcs} Hamiltonian in \cref{eq:starting Hamiltonian}, mapped onto a quantum computer as
described in \lcref{sec:Mapping on a quantum computer}, is given by
\begin{subequations}
\begin{align}
  U\pa{t} &= \ee^{-i\frac{t}{\hbar}H_{\mathrm{BCS}}}\\
          &= \pa{\prod_{q=0}^{n-1} \ee^{i\frac{t}{\hbar}g\epsilon_{q}H_{q}}}\pa{\prod_{\substack{j, k = 0\\j \neq k}}^{n-1}
        \ee^{-i\frac{t}{\hbar}g\sigma^{z}_{j}\sigma^{z}_{k}/4}}\nonumber\\
        &\eqspace\times \pa{\prod_{j=0}^{n-1} \ee^{-i\frac{t}{\hbar}g \sigma^{z}_{j} /
    2}}\,,  \label{eq:evolution splitting}
\end{align}
\end{subequations}
with the Gaudin Hamiltonions
\begin{equation}
  H_{q} = \sum_{\substack{j=0\\j\neq q}}^{n-1} \frac{\bm{\sigma}_{q} \cdot
  \bm{\sigma}_{j}}{2\left(\epsilon_{q} - \epsilon_{j}\right)} + \frac{\sigma^{z}_{q}}{g}
  \,. \label{eq:Gaudin mit sigma}
\end{equation} 
Here we made use of the fact that in the BCS Hamiltonian, \cref{eq:starting Hamiltonian},
every term commutes with all the other terms
(\lcref{sec:Properties of the Gaudin Hamiltonians}). Please note that we neglect a constant phase of $\frac{-n g}{2} + \sum_{j=0}^{n-1}
\epsilon_{j}$ with respect to the Hamiltonian in \cref{eq:base hamiltonian}. This phase
would have to be taken into account, for example, if one performs a phase estimation \cite{svore_fast_phase_estimation}
to calculate the eigenvalues of the \ac{bcs} Hamiltonian and one is interested in the absolute values of the energies.

To implement the exponential operators in \cref{eq:evolution splitting}, we define the operators
\begin{equation}
U_{\mathrm{sH}, jk}\pa{\alpha}
= \ee^{-i\alpha \bm{\sigma}_{j} \cdot \bm{\sigma}_{k}} \,, \quad U_{\mathrm{sI},
jk}\pa{\alpha} = \ee^{-i\alpha \sigma_{j}^{z} \sigma_{k}^{z}} \label{eq:evolution helpers}
\end{equation} for a parameter $\alpha \in \R$. The exponent of the first operator $U_{\mathrm{sH,
jk}}\pa{\alpha}$ describes a Heisenberg-type interaction, while the exponent of
$U_{\mathrm{sI}, jk}\pa{\alpha}$ describes an Ising-type interaction. We need to implement these two-qubit operators on the quantum processor. For this, we briefly repeat the matrix representations of some standard gates. The Pauli gates are defined, accordingly to the Pauli matrices in \cref{eq:Pauli matrices}, as
$X=\sigma_x$, $Y=\sigma_y$,
$Z=\sigma_z$.
Some useful roots of the Pauli gates are
\begin{align}
  X^{1/2} &= \frac{1}{2}\matp{1+i&1-i\\1-i&1+i} = \pa{X^{-1/2}}^{*} \,,\\
  S &= \matp{1&0\\0&i} = \sqrt{Z}\,.
\end{align}
The Hadamard gate and a rotation around the $z$-axis are given by
\begin{align}
  H &= \frac{1}{2}\matp{1&1\\1&-1} \,,\\
  R_{z}\pa{\lambda} &= \matp{\ee^{-i \lambda / 2}&0\\0&\ee^{i \lambda / 2}} \,.
\end{align}
The controlled-not gate and the swap gate are defined as
\begin{align}
  \mathrm{CNOT}_{jk} &= \matp{1&0&0&0\\0&0&0&1\\0&0&1&0\\0&1&0&0} \,,\\
  \mathrm{SWAP}_{jk} &= \matp{1&0&0&0\\0&0&1&0\\0&1&0&0\\0&0&0&1} \,,
\end{align}
where the $j\mathrm{th}$ qubit controls the $k\mathrm{th}$ qubit, assuming the basis
\begin{align}
    \ket{0}_k\ket{0}_j &= \ket{00} \leftrightarrow \pa{1, 0, 0, 0}^{T} \,,\\
    \ket{0}_k\ket{1}_j &= \ket{01} \leftrightarrow \pa{0, 1, 0, 0}^{T} \,,\\
    \ket{1}_k\ket{0}_j &= \ket{10} \leftrightarrow \pa{0, 0, 1, 0}^{T} \,,\\
    \ket{1}_k\ket{1}_j &= \ket{11} \leftrightarrow \pa{0, 0, 0, 1}^{T} \,.
\end{align}

With these gates, it is possible to implement $U_{\mathrm{sH}, jk}\pa{\alpha}$ and
$U_{\mathrm{sI}, jk}\pa{\alpha}$ in \cref{eq:evolution helpers} as shown in the \cref{qc:Implementation of the
Heisenberg time evolution operator,qc:Implementation of the Ising time evolution
operator}.
\Cref{qc:Implementation of the
Heisenberg time evolution operator} shows a general approach to construct gates with an action $\ee^{iA}$ for a
hermitian operator $A$ by implementing the basis transformation from the eigenbasis of $A$ to the $z$-basis, followed by $z$-rotations according to the eigenvalues of $A$ and a back transformation from the $z$-basis to the eigenbasis. For example, in the case of $U_{\mathrm{sH}, jk}$, the operator
$\bm{\sigma}_{j} \cdot \bm{\sigma}_{k}$ is diagonal in the Bell basis
\begin{equation}
    \ket{\phi_{\pm}} = \frac{\ket{00} \pm \ket{11}}{\sqrt{2}} \,,\quad \ket{\psi_{\pm}} = \frac{\ket{01} \pm \ket{10}}{\sqrt{2}} \,, \label{eq:bell basis}
\end{equation}
with the eigenvalue $+1$ for $\ket{\phi_{\pm}}, \ket{\psi_{+}}$ and $-3$ for $\ket{\psi_{-}}$. Therefore we map, as described in \cref{qc:Implementation of the Heisenberg time evolution operator},
\begin{subequations} (see also \cite{burkard_manuscript})
\begin{align}
& \pa{\ket{\phi_+}, \ket{\phi_-}, \ket{\psi_+}, \ket{\psi_-}}\\
&\mapsto \pa{\ket{00}, \ket{01}, \ket{10}, \ket{11}}\\
&\mapsto \pa{\ee^{-i\alpha}\ket{00}, \ee^{-i\alpha}\ket{01}, \ee^{-i\alpha}\ket{10}, \ee^{3i\alpha}\ket{11}}\\
&\mapsto \pa{\ee^{-i\alpha}\ket{\phi_+}, \ee^{-i\alpha}\ket{\phi_-}, \ee^{-i\alpha}\ket{\psi_+}, \ee^{3i\alpha}\ket{\psi_-}} \,.
\end{align}
\end{subequations}
$U_{\mathrm{sI}, jk}$ is already diagonal in the $z$-basis, so we can directly perform the $z$-rotations, as described in \cref{qc:Implementation of the Ising time evolution operator},
\begin{subequations}
\begin{align}
&\pa{\ket{00}, \ket{01}, \ket{10}, \ket{11}}\\
& \to \pa{\ee^{-i\alpha}\ket{00}, \ee^{i\alpha}\ket{01}, \ee^{i\alpha}\ket{10}, \ee^{-i\alpha}\ket{11}} \,.
\end{align}
\end{subequations}
The evolution of
\begin{equation}
  U\pa{H_{q}, t} 
  = \ee^{i\frac{t}{\hbar}g\epsilon_{q}H_{q}} \label{eq:time Gaudin}
\end{equation}
can be approximated using the Trotter-Suzuki formulas \cite{suzuki_trotter, wiebe_trotter_suzuki}, which factorize the exponential operator. The first-order and second-order Trotter-Suzuki formulas for two non-commuting operators $A$ and $B$ are given by, respectively,
\begin{align}
  \ee^{it\pa{A+B}} &= \lim_{m \to \infty}
  \pa{\ee^{iA \Delta t}\ee^{iB \Delta t}}^{m} \,, \label{eq:trotter one}\\
  \ee^{it\pa{A+B}} &= \lim_{m \to \infty}
  \pa{\ee^{iB \Delta t/2}\ee^{iA \Delta t}\ee^{iB \Delta t/2}}^{m} \,, \label{eq:trotter two}
\end{align}
with the discrete time step $\Delta t =  \frac{t}{m}$. For finite $m$ the errors $\varepsilon_{1}, \varepsilon_{2}$, for the first and second order, respectively, have the upper bound
\begin{align}
  \varepsilon_{1} &\leq \frac{t^{2}}{2 m} \norm{\bk{A, B}} + \mathcal{O}\pa{\frac{t^{3}}{m^{3}}} \label{eq:trotter_first_order} \,,\\
  \varepsilon_{2} &\leq \frac{t^{3}}{12 m^{2}} \norm{\bk{A +
  \frac{B}{2}, \bk{A, B}}} + \mathcal{O}\pa{\frac{t^{4}}{m^{4}}} \,.
\end{align}
Using the formulas in \cref{eq:trotter one,eq:trotter two}, we decompose $U\pa{H_{q}, t}$ into terms of
single rotations around the $z$-axis and $U_{\mathrm{sH}, qj}$ (\cref{eq:evolution helpers}). With this given, the time evolution $U\pa{t}$ is easily implemented on an ideal quantum register
with an all-to-all connectivity. However, such ideal quantum computers are not realistic.

We consider the case of a \ac{css} quantum register with a star-shaped connectivity map
as described in \lcref{sec:The models}. Let $q^{*}_{0}$ be the central qubit that couples to all
other qubits. To implement the time evolution, we make use of the swap gate. The algorithm is visualized in
\cref{qc:First layout}. First, we implement $U\pa{H_{0}, t}$ which only contains couplings with the central qubit $q^{*}_{0}$. Next we perform
a swap operation on the qubits $q^{*}_{0}$ and $q_{1}$. Now we can implement $U\pa{H_{1}, t}$
with adapted parameters as described in \cref{qc:First layout}b. Next we swap the states on the qubits $q^{*}_{0}$ and $q_{2}$ and
proceed in the same manner until we reach the last qubit. This procedure implements the
first term $\prod_{q=0}^{n-1}\ee^{i\frac{t}{\hbar}g\epsilon_{q}H_{q}}$ in \cref{eq:evolution splitting},
however with swapped states at the end. This will be fixed with the second term in
\cref{eq:evolution splitting}. For this, we define the operators $U\pa{I_{q}, t} =
\prod_{j = 0; j \neq q}^{n-1}\ee^{-i\frac{t}{\hbar}g\sigma^{z}_{q}\sigma^{z}_{j}/4}$. With these operators we proceed
analog as with the Gaudin terms $U\pa{H_{q}, t}$, however starting with the $(n - 1)$th qubit,
i.e., starting with $U\pa{I_{n-1}, t}$, as described in \cref{qc:First layout}c.

While the exact total number of required gates depends on the given set of native gates, the complexity, i.e., the gate count of the algorithms with respect to the number of qubits $n$ is of great interest. $U\pa{H_{q},
t}$ is implemented using the Trotter-Suzuki formula by splitting the evolution into
$r_q\pa{t, \varepsilon, \bc{\epsilon_j}, g} \cdot n$ exponential operators, where the operators are either $U_{\mathrm{sH},
jk}\pa{\alpha}$ or single qubit rotations around the $z$-axis. The factor $r_q\pa{t, \varepsilon, \bc{\epsilon_j}, g}$ depends on the chosen Trotter-Suzuki
decomposition, where $\varepsilon$ is the error of the approximation. For example, it exists a $2p$th order Trotter-Suzuki decomposition where
\begin{equation}
 r_q\pa{t, \varepsilon, \bc{\epsilon_j}, g} \approx
\mathcal{O}\pa{\frac{p 25^{p}}{3^{p-1}}\sqrt[2p]{\frac{\pa{\Lambda_q\pa{\bc{\epsilon_j}, g} t}^{2p+1}}{\varepsilon}}} \label{eq:trotter complexity}
\end{equation}
can be reached \cite{suzuki_trotter, wiebe_trotter_suzuki, papageorgiou_trotter_suzuki, raeisi_circuit_design}. For the \ac{bcs} problem, one finds that $\Lambda_q\pa{\bc{\epsilon_j}, g} \hbar = \sum_{j=0}^{n-1} \abs{\frac{3 g \epsilon_q}{2\epsilon_{q} - \epsilon_{j}}} + \abs{\epsilon_q}$. However this value is only an estimate which provides an upper bound of the number of needed gates and smaller $r_q\pa{t, \varepsilon, \bc{\epsilon_j}, g}$ may be possible. It follows that the algorithm, as proposed in \cref{qc:First layout}, has a maximum circuit-size complexity of
$\mathcal{O}\pa{\max_q\pa{r_q\pa{t, \varepsilon, \bc{\epsilon_j}, g}} n^{2}}$, both in terms of single qubit and two
qubit gates. The circuit-depth complexity is of the same order. This means that, up to the factor $\max_q\pa{r_q\pa{t, \varepsilon, \bc{\epsilon_j}, g}}$, the complexity is quadratic in the number of qubits. However, the dependence of $\max_q\pa{r_q\pa{t, \varepsilon, \bc{\epsilon_j}, g}}$ on the system parameters $\epsilon_j$ and $g$ and the time $t$ is not trivial in general.
\begin{figure}
\center
\includegraphics{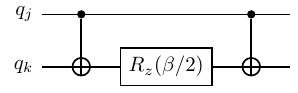}
\caption{Implementation of the Ising time
evolution operator $\ee^{-i\alpha \sigma_{j}^{z} \sigma_{k}^{z}}$. The rotation parameter is given by $\beta = 4
\alpha$. We do not need to perform a basis transformation, as in \cref{qc:Implementation of the Heisenberg time evolution operator}, since the evolution operator is already diagonal in the computational basis.}
\label{qc:Implementation of the Ising time evolution operator}
\end{figure}
\begin{figure*}[hbt!]
\center
\includegraphics{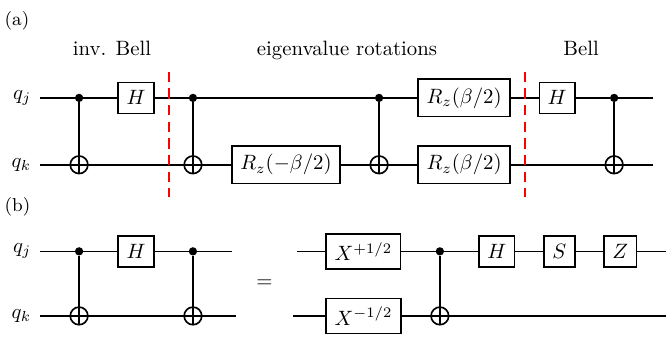}
  \captionShortcut[Implementation of the Heisenberg time evolution
  operator]{From \cite{burkard_manuscript}. (a) Implementation of the
  Heisenberg time evolution operator $\ee^{-i\alpha \bm{\sigma}_{j} \cdot
\bm{\sigma}_{k}}$. The rotation parameter is given by $\beta = 4 \alpha$. First the Bell basis, as defined in \cref{eq:bell basis} (the eigenvectors of $\bm{\sigma}_j \cdot \bm{\sigma}_k$), is mapped to the $z$ basis, then $z$-rotations, accordingly to the eigenvalues ($+1$ for $\ket{\phi_{\pm}}, \ket{\psi_{+}}$ and $-3$ for $\ket{\psi_{-}}$), are executed and in the end the $z$ basis is mapped back to the Bell basis. (b) Implementation of the first three gates in (a) to replace one $\mathrm{CNOT}$.}{qc}
\end{figure*}%

The algorithm demonstrates that a \ac{css} quantum registers represent a powerful platform when it
comes to the implementation of double products of two-qubit gates. Let us consider the operator
$f = \prod_{j \in M} \prod_{k \in S_j} f_{jk}$, where $f_{jk}$ is a unitary operator on
the qubits $j$ and $k$, for the tuples $\pa{M_j}_j, \pa{S_{j, k}}_k \subset \bc{0, \ldots, n - 1}$, where $j=1, \ldots, |M|$. This operator is a product of operators $\hat{S}_j = \prod_{k \in S_j} f_{jk}$, where each of these operators has effectively one central ``spin'' that needs to interact with all the other ``spins''. On a \ac{css}
quantum register, the operators $\hat{S}_j$ can be implemented successively by swapping the central qubit with the $j\mathrm{th}$ qubit in between and adapting the
parameters in an analog way as in \cref{qc:First layout}b. With this, the number of necessary swap gates for the implementation $f$ is of order $\mathcal{O}\pa{\abs{M}}$, where $|M|$ counts the number of times where the role of the central spin changes.
The trivial special case, $\abs{M} = 1$, can, for example, be used to simulate the central-spin system itself, which has application in solving nonlinear differential equations \cite{Geller2021}.
\begin{figure*}[hbt!]
  \center
  \includegraphics{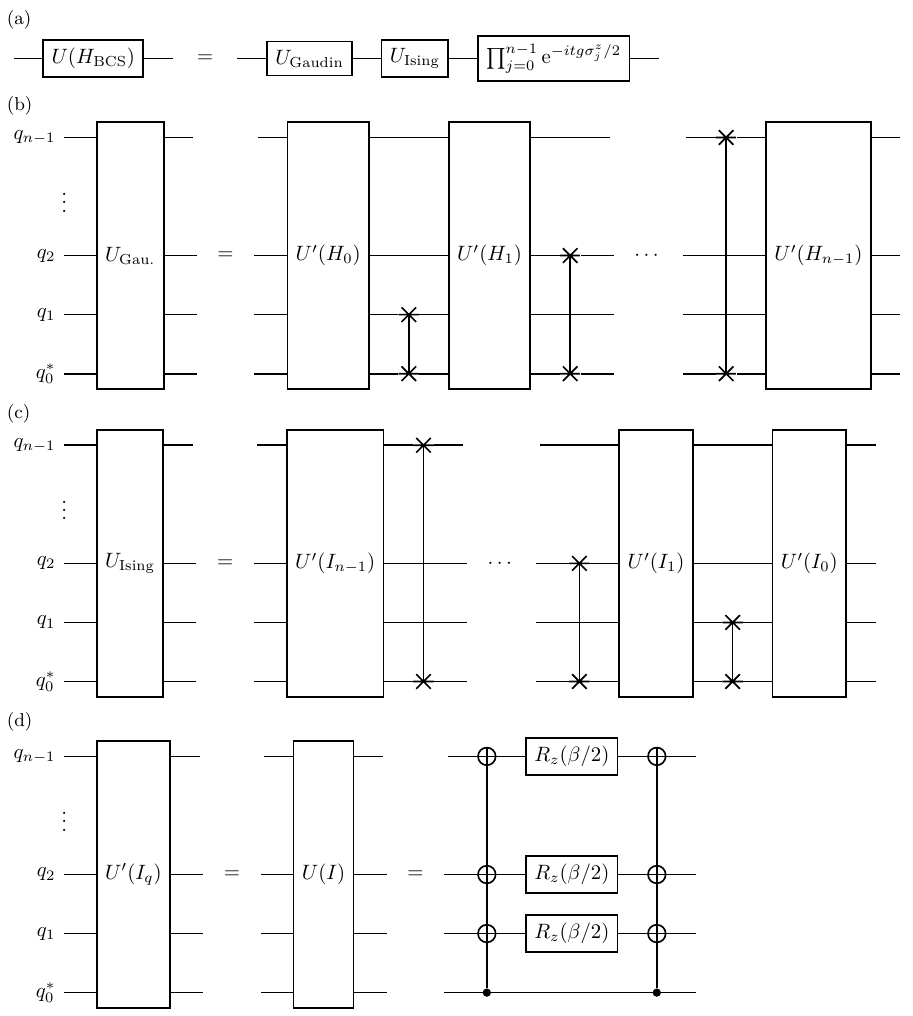}
  \captionShortcut[First layout]{Implementation of the \ac{bcs} time evolution $U\pa{t}
    = U\pa{H_{\mathrm{BCS}}}$ for constant parameters on a \ac{css} quantum register.
    The lowest wire represents the central qubit $q_{0}^{*}$. (a) Decomposition of the \ac{bcs} evolution into a Gaudin term, an Ising
    term, and some single qubit rotations. (b) Implementation of the
    Gaudin term (multiple Gaudin Hamiltonians). The operators $U'\pa{H_{j}} = U\pa{H_{0}, t}\bk{\bc{\epsilon_{j, k}}_k, g}$ are recursively defined via \cref{eq:time Gaudin} and $\bc{\epsilon_{j, k}}_k = h_{j, 0}\pa{\bc{\epsilon_{j-1, k}}_k}$ for $j > 0$ and $\bc{\epsilon_{0, k}}_k = \bc{\epsilon_{k}}_k$, where $h_{l, m}$ swaps the elements $a_l$ and $a_m$ in a tuple $\bc{a_k}$. $U\pa{H_{0}, t}$ can, for
    example, be implemented with the Trotter-Suzuki formulas.
    (c) The Ising term evolution. Note
    that this gate inverts the qubit permutations from the Gaudin gate.
    (d) Implementation of one of the Ising gates with $\beta = tg$.
    Note that the single-qubit gates do not depend on $q$.}{qc}
\end{figure*}

The algorithm for the time-independent \ac{bcs} Hamiltonian in \cref{eq:starting Hamiltonian} that we have shown above, can be easily expanded, by using the same 
 trotterization techniques that we have already used, to a more general 
time-dependent Hamiltonian including possible perturbations
\begin{equation}
  H\pa{t}= H_{\mathrm{BCS}}\pa{t} + H_{\mathrm{P}}\bk{\bc{\bm{K}_{j}}}\pa{t} \,,
\end{equation}
where the perturbation term $H_{\mathrm{P}}\bk{\bc{\bm{K}_{j}}}\pa{t}$ needs to be
expressible only using the spin operators $\bc{\bm{K}_{j}}$, so that the mapping in \cref{sec:Mapping on a quantum computer} is applicable. The time evolution operator is
given by the Dyson series
\begin{equation}
U\pa{t, t_{0}} = \mathcal{T}\bk{\ee^{-\frac{i}{\hbar} \int_{t_{0}}^{t} H\pa{t'} \diff{t'}}} \,,
\end{equation}
where $\mathcal{T}$ is the time ordering operator. To simulate the time evolution up to the time $t$, one
can discretize the total time $t - t_0$ into $m$ steps and split the time evolution operator as follows
\begin{equation}
  U\pa{t, t_{0}} = U\pa{t = t_{m}, t_{m-1}} \ldots U\pa{t_1, t_{0}} \,,
  \label{eq:discretize}
\end{equation}
where the operators $U\pa{t_{j}, t_{j-1}}$ can be approximated with
\begin{equation}
U\pa{t_{j}, t_{j-1}} \approx \ee^{-\frac{i}{\hbar} H\pa{t_{j-1}} \Delta
t_{j, j-1}} \,,
\end{equation}
if the chosen time difference $\Delta t_{j, j-1} = t_{j} - t_{j-1}$ is sufficiently small. The operators
$U\pa{t_{j}, t_{j-1}}$ can be approximated using the Trotter-Suzuki
decompositions and the
implementation for the \ac{bcs} evolution from the constant case.

\section{Results} \label{sec:Results}

In this section we present our numerical results for the simulation of the dynamics of the model and discuss further optimization strategies and application fields. The numerical calculations are performed with a simulated quantum computer. As a test for the proper function of the quantum simulation, we calculate the return probability, i.e, the probability  that the system after time $t$ (described by the state vector $\ket{\psi\pa{t}}$) has  returned to its initial state $\ket{\psi_0}$,
\begin{equation}
  \abs{\braket*{\psi_0}[\psi\pa{t}]}^{2} =
  \abs{\braket*{\psi_0}[\mathcal{T}\bk{\ee^{-\frac{i}{\hbar} \int_{t_0}^{t}
  H_{\mathrm{BCS}}\pa{t'} \diff{t'}}}][\psi_0]}^{2} \,. \label{eq:simulation quantity}
\end{equation}
Note that the return probability  equals the Loschmidt echo, which constitutes an important quantity in multiple contexts of the quantum many-body theory, for example, quantum chaos and nonequilibrium fluctuation theorems \cite{Heyl2018,adamov_loschmidt_echo}. 
Here, we assume a Hamiltonian $H_{\mathrm{BCS}}\pa{t}$ as in \cref{eq:starting Hamiltonian} with time-dependent parameters. As initial state, we use 
$\ket{\psi_0} = \ket{\mathrm{BCS}}$, where $\ket{\mathrm{BCS}}$ is the ground state of the mean-field
\ac{bcs} theory at time $t_0$, given as
\begin{equation}
  \ket{\mathrm{BCS}} = \prod_{j=0}^{n-1} \pa{u_{j} - v_{j} K_{j}^{+}} \ket{0} \,. \label{eq:BCS state}
\end{equation}
The parameters $u_j, v_j \in \mathbb{C}$ depend on the system parameters $\epsilon_k\pa{t_0}$ and $g\pa{t_0}$,
for details see \lcref{sec:Mean-field ground state} and \cite{tinkham_introduction_to_superconductivity}. Since $\ket{\mathrm{BCS}} \in \mathcal{H}_{\mathrm{C}}$, the problem is suitable for the algorithm presented above. For a constant Hamiltonian the state $\ket{\mathrm{BCS}}$ approximates the ground state for $n \to \infty$, implying that the return probability in \cref{eq:simulation quantity} approaches $1$. Because of its form of a product state, $\ket{\mathrm{BCS}}$ can be easily implemented using single qubit rotations.

In the presence of errors, a quantum simulation is not perfect; rather the simulation results in a mixed state, which can be described with a density matrix $\rho\pa{t}$. Therefore, instead of the return probability as in \cref{eq:simulation quantity}, we actually calculate
\begin{equation}
    \mathcal{R}_{\mathrm{mf}}\pa{t} = \braket*{0}[\rho_{\mathrm{mf}}\pa{t}][0] \,. \label{eq:mf return rate}
\end{equation}
The density matrix $\rho_{\mathrm{mf}}$ is the result of the quantum simulation, which consists of initializing the mean-field ground state $\ket{\mathrm{BCS}}$, performing the time evolution and inverting the mean-field ground state initialization. All these operations might be error-prone. In the optimal case, without any errors, the density matrix describes the following pure state
\begin{align}
    \rho_{\mathrm{mf}}^{\mathrm{opt}}\pa{t} &= \ketbra{\phi_{\mathrm{mf}}\pa{t}}[\phi_{\mathrm{mf}}\pa{t}] \,, \label{eq:mf opt a}\\
    \ket{\phi_{\mathrm{mf}}\pa{t}} &= \braket*{\psi_{\mathrm{mf}, 0}}[\psi_{\mathrm{mf}}\pa{t}] \ket{0} \nonumber\\
    &\eqspace + \sqrt{1 - \abs{\braket*{\psi_{\mathrm{mf}, 0}}[\psi_{\mathrm{mf}}\pa{t}]}^2} \ket{0^{\perp}} \,, \label{eq:mf opt b}
\end{align}
where $\ket{0^{\perp}}$ is a state orthogonal to $\ket{0}$. $\ket{\psi_{\mathrm{mf}, 0}}$ and $\ket{\psi_{\mathrm{mf}}\pa{t}}$ are the state $\ket{\mathrm{BCS}}$ and its time evolved state, respectively. In this optimal case, $\mathcal{R}_{\mathrm{mf}}\pa{t}$ equals the formula in \cref{eq:simulation quantity}.

In addition to the simulation of the mean-field ground state, we calculate the return probability for the exact ground state of the Hamiltonian in \cref{eq:starting Hamiltonian}. However, we do not implement the initialization of this state in the quantum algorithm; instead, we directly specify this state as initial state. This is only possible because we use a simulated quantum computer and not a real quantum device. The resulting quantity of the simulation is the return probability
\begin{equation}
    \mathcal{R}_{\mathrm{exact}}\pa{t} = \braket*{\psi_{\mathrm{exact}, 0}}[\rho_{\mathrm{exact}}\pa{t}][\psi_{\mathrm{exact}, 0}] \,, \label{eq:exact return rate}
\end{equation}
where, in the optimal case, without qubit and gate errors, the density matrix $\rho_{\mathrm{exact}}\pa{t}$ describes the state
\begin{equation}
    \rho_{\mathrm{exact}}^{\mathrm{opt}}\pa{t} = \ketbra{\psi_{\mathrm{exact}}\pa{t}}[\psi_{\mathrm{exact}}\pa{t}] \,. \label{eq:exact opt}
\end{equation}
$\ket{\psi_{\mathrm{exact}, 0}}$ and $\ket{\psi_{\mathrm{exact}}\pa{t}}$ are the exact ground state, at time $t_0$, and its time evolved state, respectively.
\begin{figure*}[hbt!]
\center
\includegraphics{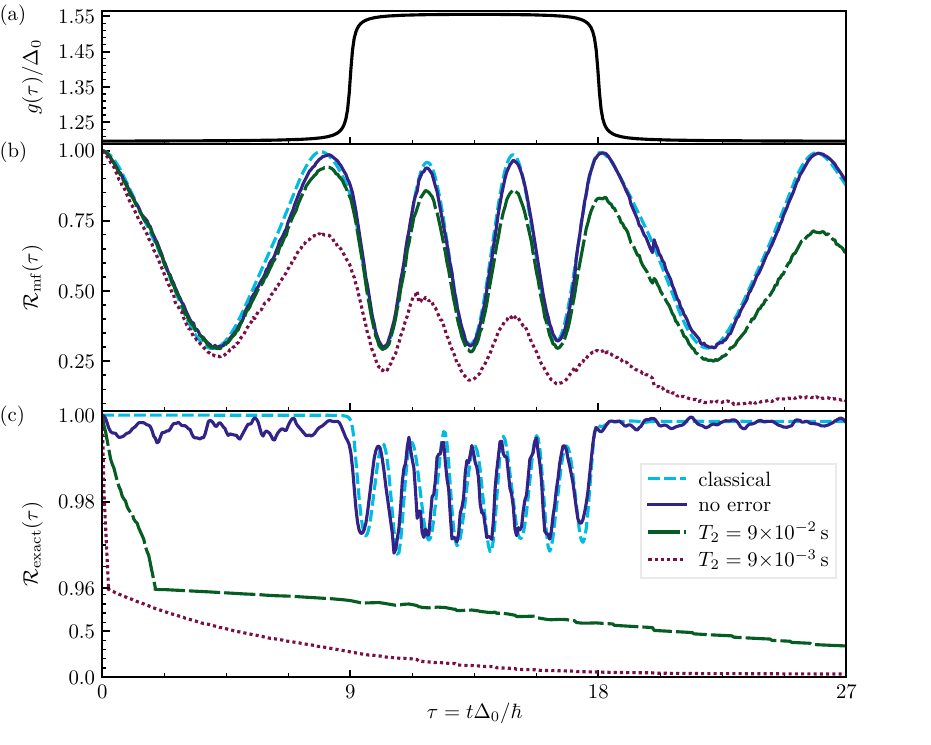}
\caption{The time dependent return probability $\mathcal{R}\pa{t}$ of a state $\ket{\psi\pa{t}}$ to its initial state $\ket{\psi_0}$ ($t_0 = 0$), where the state $\ket{\psi\pa{t}} = \mathcal{T}\bk{\exp\pa{-\frac{i}{\hbar} \int_{t_0}^{t}
  H_{\mathrm{BCS}}\pa{t'} \diff{t'}}} \ket{\psi_0}$ describes the time evolution, determined by the \ac{bcs} Hamiltion in \cref{eq:starting Hamiltonian}. (a) The time dependent coupling function $g\pa{\tau} \coloneqq g\pa{t\pa{\tau}}$ as defined in \cref{eq:coupling}. (b) The return probability $\mathcal{R}_{\mathrm{mf}}\pa{t} \coloneqq \mathcal{R}_{\mathrm{mf}}\pa{t\pa{\tau}}$, described in \cref{eq:mf return rate,eq:mf opt a,eq:mf opt b}, for the mean-field ground state $\ket{\mathrm{BCS}}$, defined in \cref{eq:BCS state}, as initial state. (c) The return probability $\mathcal{R}_{\mathrm{exact}}\pa{t} \coloneqq \mathcal{R}_{\mathrm{exact}}\pa{t\pa{\tau}}$, as described in \cref{eq:exact return rate,eq:exact opt} for the exact ground state of the \ac{bcs} Hamiltonian in \cref{eq:starting Hamiltonian} as initial state. In both plots, (b) and (c), the results are calculated with a simulated quantum computer provided by \cite{qiskit_github}. The dark blue solid line shows the return probability, calculated with the quantum algorithm under the assumption of error-free qubits and gates, and the light blue dashed line is the classical calculated return probability (``classical'' in the sense that a classical algorithm with high precision is used). Up to numerical errors, caused by the trotterization, these lines are the same. The green long-dashed line and the red-magenta colored dotted line are the results of the quantum simulation with noisy qubits, but without gate errors, i.e., the execution of the gates is assumed to be error free, however, the qubit errors can still spread from one qubit to another. As for the qubit error, we model transversal and longitudinal relaxation based on an amplitude-phase-damping channel with the coherence times $T_1 = \SI{1.25e-1}{s}$, $T_2 = \SI{9e-2}{s}$ (green long-dashed), $T_2 = \SI{9e-3}{s}$ (magenta-red dotted), a single qubit gate time $t_s = \SI{5e-8}{s}$ and a two qubit gate time $t_t = \SI{5e-7}{s}$. We do not consider any coherent or cross-talk errors.
}
\label{fig:plotRectangle}
\end{figure*}

A quantum quench describes the process of initializing a system in a certain state, often an eigenstate, e.g., the ground state, and subjecting the system to a time dependent modification of parameters or, for example, a perturbation \cite{Heyl2018,mitra_quantum_quench_dynamics,calabrese_ising_quench}. We simulate a quench, varying the superconducting gap $\Delta$, realized via a change of the coupling constant $g$. After some time the quench is performed backwards, i.e., $g$ is reset to its initial value.

We introduce a
dimensionless time $\tau\pa{t} = \frac{t \mathfrak{J}}{\hbar}$, where $\mathfrak{J}$
is an arbitrary energy unit. Without loss of generality, we set $t_0 = 0$. The classical simulation of our quantum circuit is
done for $n = 5$ qubits. For the energy levels we choose a harmonic oscillator, i.e.,
$\epsilon_j = \omega \pa{j + \frac{1}{2}}$, as one of the simplest non-interacting systems. The coupling strength $g$ is time dependent, according to
\begin{align} \label{eq:coupling}
  g\pa{t} &= \frac{\pa{g_c -
g_0}}{\pi^{2}} \bk{\arctan\pa{\pa{t - t_1}
  \frac{\mathfrak{J}}{\hbar\Gamma}} + \frac{\pi}{2}}\nonumber\\
  &\eqspace \times \bk{\arctan\pa{\pa{t_2 - t}
\frac{\mathfrak{J}}{\hbar\Gamma}} + \frac{\pi}{2}} +
g_0 \,,
\end{align}
which is plotted in \cref{fig:plotRectangle}a. The parameter $\Gamma$ describes the smoothness of the quench and $t_1$ and $t_2$ are the times when the quench and the reverse quench take place, respectively. $g_0$ is the initial coupling constant and $g_c$ is the coupling constant after the quench. The results of the numerical simulation are depicted in \cref{fig:plotRectangle}. The chosen set of parameters is given by $t_1 = 9 \, \hbar/\mathfrak{J}$, $t_2 = 18 \, \hbar/\mathfrak{J}$, $\Gamma = 0.1$ and $\omega = \frac{5}{3} \, \mathfrak{J}$, while $g_0$ and $g_c$ are
calculated from the superconducting gaps $\Delta_0 = \mathfrak{J}$ and $\Delta_c
= 2 \,\mathfrak{J}$, respectively (details are given in
\lcref{sec:Mean-field ground state}, \cref{eq:g_from_delta}). For all our simulations, we remain at zero temperature, $T = \SI{0}{K}$. The trotterization of the Gaudin Hamiltonians is performed using the first and second order equations from \cref{eq:trotter one,eq:trotter two}, where we specified the error $\varepsilon_1$ in \cref{eq:trotter_first_order} to be smaller than $\frac{3}{5} C$, where $C$ is the constant factor caused by the non-commuting terms in the Gaudin Hamiltonians, i.e., we set the number of Trotter steps to $m\pa{\tau} \approx \frac{5}{6} \tau^2$. This is only an approximation because the trotter-step-width $\tau/m\pa{\tau}$ has to be adapted to the splitting of the Dyson series, which depends dynamically on the system parameters in our simulation. With that we can count the number of $\cnot$s in our quantum circuit: There are $2\pa{n-1}$ $\Swap$ gates where each can be decomposed into $3$ alternating $\cnot$s. We have $n$ Gaudin terms, where each of them is trotterized with $m\pa{\tau}$ steps; each step contains of $n-1$ Heisenberg evolution operators as in \cref{qc:Implementation of the Heisenberg time evolution
  operator} (first order trotterization), which require $3$ $\cnot$s. There are $n$ Ising terms and each of them contains $n-1$ Ising evolution operators as in \cref{qc:Implementation of the Ising time evolution operator} with $2$ $\cnot$s. Summing things up, the total number of $\cnot$s is
\begin{subequations}
\begin{align}
    N_{\cnot}\pa{\tau} &= 6(n-1) + 3n\pa{n-1}m\pa{\tau} + 2n\pa{n-1}\\
    &= 2n^2 + 4n - 1 + \frac{5}{2}\pa{n^2 - n}\tau^2\\
    &= 69 + 50 \tau^2
\end{align}
\end{subequations}
In the last step we substituted $n = 5$. If we insert the largest simulation time in \cref{fig:plotRectangle}, $\tau = 27$, we have $36519$ $\cnot$ gates. Similar counting can be performed for for the single qubit gates.

\Cref{fig:plotRectangle} shows the results of our simulations. The simulations are performed with and without qubit errors, however always with perfect gates. The qubit errors are modelled with an amplitude-phase-damping channel. To compare the results, we additionally plotted the results from a classical algorithm, which is based on the diagonalization of the Hamiltonian at multiple time steps. The perfect quantum simulations, without qubit errors, lead to the same results as the classical algorithm, up to trotterization errors.

In the plotted regime, the mean-field ground state is apparently not a good approximation of the exact ground state, but this is not unexpected since we only consider five energy orbitals. \Cref{fig:fidelity} shows the fidelity between the mean-field ground state and the exact ground state as function of the orbital number $n$ for the chosen system parameters. For $n = 5$ the fidelity is approximately $0.217363$. The very small gradient of the fidelity in \cref{fig:fidelity} indicates that the mean-field approximation does not perform very well for the chosen system parameters regarding the approximation the exact ground state. This may be partly explained by the fact that the $\ket{\mathrm{BCS}}$ state can be considered as solution of a variation ansatz minimizing the energy expectation value. This means that while the energy expectation value of $\ket{\mathrm{BCS}}$ approximates the ground state energy fairly well, the state itself may not approximate the ground state similarly well if there is some other eigenstate that has an energy near to the ground state energy.  This justifies the use of the exact ground state in our simulations.

The curves that are simulated with noisy qubits deviate strongly from the perfect simulation. These deviations increase with time $t$ since more gates are needed and therefore the duration of the computation increases. This increases the effect of the qubit relaxation errors. In the case of the exact ground state, the relative differences between the extrema are so small that it is difficult to resolve any qualitative behavior if we consider the qubit noise.
\begin{figure}[t]
\center
\includegraphics{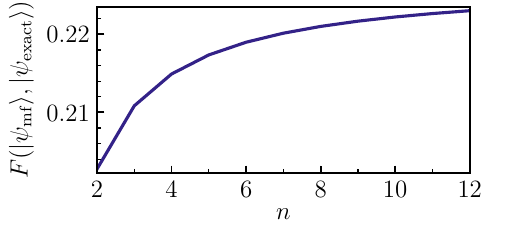}
\caption{
The fidelity $F$ between the mean-field ground state $\ket{\psi_{\mathrm{mf}}}$ and the exact ground state $\ket{\psi_{\mathrm{exact}}}$ as function of the number of orbitals $n$. Note that for pure states the fidelity equals the squared absolute value of the scalar product, $F=\abs{\braket{\psi_{\mathrm{mf}}}[{\psi_{\mathrm{exact}}}]}^2$. In this case $F$ is time-independent if the states are subject to the same time evolution.
}
\label{fig:fidelity}
\end{figure}

To improve the algorithmic performance for longer times $t$, one has to minimize the number of needed gates. 
One approach is to optimize the choice of the time steps in \cref{eq:discretize}.
In areas where $g\pa{t} \approx \text{const}$, the time step $\Delta t_{j, j-1}$ can
be bigger than in areas where $g\pa{t}$ is changing fast. We chose the time-steps
depending only on the first and second derivative of $g\pa{t}$ in a similar manner to
gradient descent methods, i.e., we made a ``big'' time step if both, the first and second derivatives were ``small'' and vice versa. However this approach does not directly reduce the number of needed gates for the trotterization. This may be reached by using higher order Trotter-Suzuki formulas,
however please note that the optimal order depends on the time, e.g., as in \cref{eq:trotter complexity}, and it is even more difficult to find the optimal order if the system parameters vary over time. In the present paper, we used first and
second order formulas as given in the \cref{eq:trotter one,eq:trotter two}. Another more hardware-specific optimization would be
to implement the circuit using only native gates and if possible using the ones with the
smallest errors, i.e., helping the transpiler to find the best circuit. One can also try
to trotterize $U\pa{H_{q}, t}$ into gates including more than two-qubits. We used the two-qubit
gate described in \cref{qc:Implementation of the Heisenberg time evolution operator} and
a rotation around the $z$-axis.

\section{Conclusion}

Our work provides a quantum algorithm capable of simulating the time-dependent \ac{bcs} model. We restricted ourselves to the space of Cooper-pairs, which enabled us to map the physical problem very efficiently with order $\mathcal{O}\pa{1}$ onto a quantum register, in contrast, for example, to the Jordan-Wigner mapping. The algorithm exploits invariants of the \ac{bcs} system, i.e., we expressed the Hamiltonian with the commuting Gaudin Hamiltonians. Furthermore, we used the structure of the Gaudin Hamiltonians to implement the algorithm on a quantum register with a star-shaped coupling map, making only use of its native connectivity. Additionally, we demonstrated that this algorithm provides a general effective method to implement double products of two-qubit operators on such a quantum register. Finally, we showed some numerical results, simulating a quenched time evolution of the mean-field ground state and proposed possible optimizations for future work. Further alternative methods, which might be interesting to improve the performance of the algorithm, such as simulating the time evolution via truncating the Taylor series of the time evolution exponential, are proposed in \cite{childs_toward_the_first_simulation, berry_alternative_simulation}. The simulation results with quantum errors indicate that quantum error correction and/or better quantum hardware will be needed to perform real quantum simulations with valuable results. E.g., it has been shown that crosstalk errors can be mitigated with an appropriate algorithm \cite{perrin_crosstalk}. Apart from simulating the time evolution, a possible extension of our proposed algorithm is the calculation of the eigenvalues of the \ac{bcs} Hamiltonian via (hybrid) quantum phase estimation \cite{svore_fast_phase_estimation}, which makes use of a controlled time evolution (\lcref{sec:Controlled time evolution}).

\section*{Acknowledgments}
We acknowledge funding from the state of Baden-Württemberg through the Kompetenzzentrum Quantum Computing, Project QC4BW.

\onecolumngrid 
\appendix

\section{Properties of the Gaudin Hamiltonians} \label{sec:Properties of the Gaudin Hamiltonians}

If not otherwise noted, sums over Latin indices (e.g., $j, k, p, q$) run from $0$ to $n-1$
while sums over Greek indices (e.g., $\alpha, \beta, \gamma$) assume the values $x, y, z$.
\begin{theorem}  \label{the:commute with each other}
The Gaudin Hamiltonians commute with each other, i.e., $\bk{H_{q}, H_{p}} = 0$ {\upshape\cite{frenkel_gaudin_commutator_proof}}.
\end{theorem}
\begin{proof}
Let $q \neq p$. We calculate the commutator separated in three steps. Let us start with the
commutator of the last terms and the commutator of the mixed terms: 
\begin{align}
  \bk{K^{z}_{q}, K^{z}_{p}} &= 0 \,,
\end{align}
\vspace{-0.3cm}
\begin{subequations}
\begin{align}
  \bk{K^{z}_{q}, \sum_{j\neq p} \frac{\bm{K}_{p} \cdot \bm{K}_{j}}{\epsilon_{p} -
\epsilon_{j}}} +  \bk{\sum_{k\neq q} \frac{\bm{K}_{q} \cdot \bm{K}_{k}}{\epsilon_{q} -
\epsilon_{k}}, K^{z}_{p}} &= \bk{K^{z}_{q}, \frac{\bm{K}_{p} \cdot
\bm{K}_{q}}{\epsilon_{p} - \epsilon_{q}}} + \bk{\frac{\bm{K}_{q} \cdot
\bm{K}_{p}}{\epsilon_{q} - \epsilon_{p}}, K^{z}_{p}}\\
&= \frac{1}{\epsilon_{p} -
\epsilon_{q}} \bk{K^{z}_{q} + K^{z}_{p}, \bm{K}_{p} \cdot \bm{K}_{q}}\\
  &= 0 \,.
\end{align}
\end{subequations}
Now the commutator of the first terms (the intermediate steps are explained below):
\begin{subequations}
\begin{align}
  \bk{\sum_{j\neq p} \frac{\bm{K}_{p} \cdot
\bm{K}_{j}}{\epsilon_{p} - \epsilon_{j}}, \sum_{k\neq q}
\frac{\bm{K}_{q} \cdot \bm{K}_{k}}{\epsilon_{q} - \epsilon_{k}}} &= \sum_{\alpha\beta}
\sum_{j\neq p} \sum_{k\neq q}
\frac{\bk{K^{\alpha}_{p} K^{\alpha}_{j}, K^{\beta}_{q} K^{\beta}_{k}}}{\pa{\epsilon_{p} -
\epsilon_{j}}\pa{\epsilon_{q} - \epsilon_{k}}}\label{eq:com_a}\\
&= \sum_{\alpha\beta\gamma} i \varepsilon^{\alpha\beta\gamma} \bcdot{
\sum_{j\neq p, q} \frac{K^{\alpha}_{p} K^{\beta}_{q} K^{\gamma}_{j}}{\pa{\epsilon_{p} -
\epsilon_{j}}\pa{\epsilon_{q} - \epsilon_{j}}} }\nonumber\\
    &\mathrel{\phantom{=}}+ \dotbc{ \sum_{j\neq q} \frac{K^{\alpha}_{p} K^{\gamma}_{q}
  K^{\beta}_{j}}{\pa{\epsilon_{p} - \epsilon_{q}}\pa{\epsilon_{q} - \epsilon_{j}}} +
  \sum_{j\neq p} \frac{K^{\beta}_{q} K^{\gamma}_{p}
  K^{\alpha}_{j}}{\pa{\epsilon_{p} - \epsilon_{j}}\pa{\epsilon_{q} - \epsilon_{p}}}}\label{eq:com_b}\\
&= \sum_{\alpha\beta\gamma}  \sum_{j\neq p, q} i
\varepsilon^{\alpha\beta\gamma} \bcdot{ \frac{K^{\alpha}_{p}K^{\beta}_{q}
K^{\gamma}_{j}}{\pa{\epsilon_{p} - \epsilon_{j}}\pa{\epsilon_{q} - \epsilon_{j}}} }\nonumber\\
    &\mathrel{\phantom{=}}+ \dotbc{  \frac{K^{\alpha}_{p} K^{\gamma}_{q} K^{\beta}_{j}}{\pa{\epsilon_{p} -
    \epsilon_{q}}\pa{\epsilon_{q} - \epsilon_{j}}} + \frac{K^{\beta}_{q} K^{\gamma}_{p}
K^{\alpha}_{j}}{\pa{\epsilon_{p} - \epsilon_{j}}\pa{\epsilon_{q} - \epsilon_{p}}}}\label{eq:com_c}\\
&= \sum_{\alpha\beta\gamma}  \sum_{j\neq p, q} i
\varepsilon^{\alpha\beta\gamma} K^{\alpha}_{p}K^{\beta}_{q}
K^{\gamma}_{j}\bcdot{ \frac{1}{\pa{\epsilon_{p} - \epsilon_{j}}\pa{\epsilon_{q} - \epsilon_{j}}} }\nonumber\\
    &\mathrel{\phantom{=}}- \dotbc{  \frac{1}{\pa{\epsilon_{p} -
    \epsilon_{q}}\pa{\epsilon_{q} - \epsilon_{j}}} - \frac{1}{\pa{\epsilon_{p} - \epsilon_{j}}\pa{\epsilon_{q} - \epsilon_{p}}}}\label{eq:com_d}\\
  &= 0 \,.\label{eq:com_e}
\end{align}
\end{subequations}
In \cref{eq:com_a}~$\to$~\cref{eq:com_b} we used
\begin{subequations}
\begin{align}
  \bk{K^{\alpha}_{p} K^{\alpha}_{j}, K^{\beta}_{q} K^{\beta}_{k}} &= K^{\alpha}_{p}
K^{\beta}_{q} \bk{K^{\alpha}_{j}, K^{\beta}_{k}} + K^{\alpha}_{p} \bk{K^{\alpha}_{j}, K^{\beta}_{q}} K^{\beta}_{k}\nonumber\\
  &\mathrel{\phantom{=}}+ K^{\beta}_{q} \bk{K^{\alpha}_{p}, K^{\beta}_{k}} K^{\alpha}_{j} +
  \bk{K^{\alpha}_{p}, K^{\beta}_{q}} K^{\beta}_{k} K^{\alpha}_{j}\\
  &= \sum_{\gamma} i \varepsilon^{\alpha\beta\gamma} \pa{\delta_{j,k} K^{\alpha}_{p}
  K^{\beta}_{q} K^{\gamma}_{j} + \delta_{j,q} K^{\alpha}_{p} K^{\gamma}_{j}
K^{\beta}_{k} + \delta_{p,k} K^{\beta}_{q} K^{\gamma}_{k} K^{\alpha}_{j}} \,,
\end{align}
\end{subequations}
with the Levi-Civita symbol, $\epsilon^{\alpha\beta\gamma} = 1$ if $\pa{\alpha, \beta, \gamma} = \pa{x, y, z}$, cyclical, and otherwise $\epsilon^{\alpha, \beta, \gamma} = -1$. In \cref{eq:com_b}~$\to$~\cref{eq:com_c} we used that in the last two sums the terms for $j=p$ and $j=q$, respectively,
cancel each other, since for $|\{\alpha, \beta, \gamma\}| = 3$ (fix $\gamma$ and exchange $\alpha$ and $\beta$)
\begin{equation}
  K^{\alpha}_{p} K^{\gamma}_{q} K^{\beta}_{p} + K^{\beta}_{q} K^{\gamma}_{p}
  K^{\alpha}_{q} - K^{\beta}_{p} K^{\gamma}_{q} K^{\alpha}_{p} - K^{\alpha}_{q}
  K^{\gamma}_{p} K^{\beta}_{q} \propto K^{\gamma}_{q}K^{\gamma}_{p} -
  K^{\gamma}_{p}K^{\gamma}_{q} = 0 \,.
\end{equation}
\cref{eq:com_c}~$\to$~\cref{eq:com_d} follows from permuting the indices $\alpha, \beta, \gamma$ and adapting the
signs. Finally, \cref{eq:com_d}~$\to$~\cref{eq:com_e} is valid, since the term in the braces equals zero.
\end{proof}
\begin{lemma}  \label{lem:propto L}
The sum of the Gaudin Hamiltonians is proportional to the $z$-component of the total
angular momentum, i.e., $-\gamma L^{z} = \sum_{q} H_{q}$ {\upshape\cite{sklyanin_gaudin_models}}.
\end{lemma}
\begin{proof}
\begin{subequations}
\begin{align}
  \sum_{q} H_{q} + \gamma L^{z} &= 2 \sum_{q} \sum_{j\neq q} \frac{\bm{K}_{q} \cdot \bm{K}_{j}}{\epsilon_{q} - \epsilon_{j}}\\
&= \sum_{\substack{q,j\\j\neq q}} \frac{\bm{K}_{q} \cdot \bm{K}_{j}}{\epsilon_{q} -
    \epsilon_{j}} - \sum_{\substack{q,j\\j\neq q}} \frac{\bm{K}_{j} \cdot \bm{K}_{q}}{\epsilon_{j} - \epsilon_{q}}\\
  &= 0\\&\nonumber\qedhere
\end{align}
\end{subequations}
\end{proof}
\begin{proposition}  \label{cor:commute with L}
The Gaudin Hamiltonians commute with the $z$-component of the total angular momentum,
i.e., $\left[H_{q}, L^{z}\right] = 0$.
\end{proposition}
\begin{proof}
Follows directly from \cref{the:commute with each other,lem:propto L}.
\end{proof}
\begin{theorem}  \label{the:construct bcs}
For $\gamma = -\frac{2}{g}$ one can construct the \ac{bcs} Hamiltonian with the Gaudin
Hamiltonians {\upshape\cite{cambiaggio_integrability_pairing}}: 
\begin{equation}
  H_{\mathrm{BCS}} = -g \sum_{q}\epsilon_{q} H_{q} + g L^{z} + g \pa{L^{z}}^{2} + \mathrm{const} \,. \label{eq:splitted bcs}
\end{equation}
\end{theorem}
\begin{proof}
It is
\begin{equation}
  \sum_{q}\epsilon_{q} H_{q} = \bm{L}^{2} -
  \sum_{q}\pa{\bm{K}^{2}_{q} + \gamma \epsilon_{q} K^{z}_{q}} \,,
\end{equation}
since
\begin{subequations}
\begin{align}
  2 \sum_{\substack{q,j\\j\neq q}} \epsilon_{q} \frac{\bm{K}_{q} \cdot
  \bm{K}_{j}}{\epsilon_{q} - \epsilon_{j}} &= \sum_{\substack{q,j\\j\neq q}}
  \frac{\epsilon_{q} \bm{K}_{q} \cdot \bm{K}_{j}}{\epsilon_{q} - \epsilon_{j}} -
  \sum_{\substack{q,j\\j\neq q}} \frac{\epsilon_{q} \bm{K}_{j} \cdot
  \bm{K}_{q}}{\epsilon_{j} - \epsilon_{q}}\\
&= \sum_{\substack{q,j\\j\neq q}} \frac{\epsilon_{q} \bm{K}_{q} \cdot
    \bm{K}_{j}}{\epsilon_{q} - \epsilon_{j}} - \sum_{\substack{q,j\\j\neq q}}
      \frac{\epsilon_{j} \bm{K}_{q} \cdot \bm{K}_{j}}{\epsilon_{q} - \epsilon_{j}}\\
  &= \sum_{\substack{q,j\\j\neq q}} \bm{K}_{q} \cdot \bm{K}_{j} \,.
\end{align}
\end{subequations}
Splitting the $\bm{L}^{2}$ we obtain
\begin{equation}
  \sum_{q}\epsilon_{q} H_{q} = \frac{\gamma}{2}\pa{-\sum_{q} 2
  \epsilon_{q} K^{z}_{q} + \frac{2}{\gamma} L^{+}L^{-}} + L^{z} + \pa{L^{z}}^{2} -
  \sum_{q} \bm{K}^{2}_{q} \,,
\end{equation}
where we can identify the term in the parenthesis with the \ac{bcs} Hamiltonian.
\end{proof}
\begin{proposition}
The Gaudin Hamiltonians and the $z$ component of the total angular momentum commute with the \ac{bcs} Hamiltonian, meaning
$\bk{H_{\mathrm{BCS}}, H_{q}} = 0$ and $\bk{H_{\mathrm{BCS}}, L^{z}} = 0$. All in all, all terms on the right-hand-side in \cref{eq:splitted bcs} commute with each other.
\end{proposition}
\begin{proof}
Follows directly from \cref{the:commute with each other,the:construct bcs,cor:commute with L}.
\end{proof}

\section{Mean-field ground state} \label{sec:Mean-field ground state}

The mean-field ground state is obtained by inserting the approximation
\begin{subequations}
\begin{align}
  c^{\dagger}_{j\up} c^{\dagger}_{j\down} c_{k\down} c_{k\up} &\approx
  \braket*{c^{\dagger}_{j\up} c^{\dagger}_{j\down}} \braket*{c_{k\down} c_{k\up}} +
  \braket{c^{\dagger}_{j\up} c^{\dagger}_{j\down}} c_{k\down} c_{k\up} +
  c^{\dagger}_{j\up} c^{\dagger}_{j\down} \braket*{c_{k\down} c_{k\up}}\\
  &= \Delta_{j}^{*} \Delta_{k} + \Delta_{j}^{*} c_{k\down} c_{k\up} +
  c^{\dagger}_{j\up} c^{\dagger}_{j\down} \Delta_{k}
\end{align}
\end{subequations}
in the Hamiltonian in \cref{eq:base hamiltonian} and diagonalizing the resulting
Hamiltonian with a Boguliubov transformation. $\Delta_{j} = - \sum_k V_{jk} \braket{c_{k\down} c_{k\up}}$ is the superconducting gap for each energy level $j$ where we replaced the constant coupling strength $-g$ with $V_{jk}$. Without going into more detail \cite{tinkham_introduction_to_superconductivity} we present the resulting
ground state:
\begin{equation}
  \ket{\mathrm{BCS}} = \prod_{j} \pa{u_{j} - v_{j} K_{j}^{+}} \ket{0} \,,
\end{equation}
with
\begin{align}
  \abs{u_j}^{2} &= \frac{1}{2} \pa{1 + \frac{\epsilon_j}{E_j}} \,,\\
  \abs{v_j}^{2} &= \frac{1}{2} \pa{1 - \frac{\epsilon_j}{E_j}} \,,\\
  \frac{v_j \Delta_{j}^{*}}{u_j} &= E_j - \epsilon_j \in \mathbb{R}_{+}\,,
\end{align}
where we used the mean-field eigenvalues
\begin{equation}
  E_j = \sqrt{\epsilon_j^{2} + \abs{\Delta_j}^{2}} \,.
\end{equation}
The superconducting gaps must fulfill the system of gap equations
\begin{equation}
\Delta_j = - \sum_{k} V_{jk} \frac{\Delta_k}{2 E_k} \tanh\pa{\frac{E_k}{2
  k_{\mathrm{B}} T}} \,, \label{eq:gap equation}
\end{equation}
for $j \in \bc{1, \ldots, n-1}$, where $T$ is the temperature and $k_{\mathrm{B}}$ the Boltzmann-constant. For $V_{jk} =
-g$, the right-hand side in \cref{eq:gap equation} is
independent of $j$, which implies
\begin{equation}
  \Delta_j = \Delta
\end{equation}
and for $\Delta \neq 0$
\begin{equation}
  \frac{2}{g} = \sum_{k} \frac{1}{E_k} \tanh\pa{\frac{E_k}{2
  k_{\mathrm{B}} T}} \,. \label{eq:g_from_delta}
\end{equation}

\section{Controlled time evolution} \label{sec:Controlled time evolution}

The algorithm described in \lcref{sec:Simulation} can be extended to a controlled version. Adding an additional control qubit $\ket{\phi}$, the time evolution shall be executed if $\ket{\phi} = \ket{1}$ and not executed if $\ket{\phi} = \ket{0}$. This can be reached by controlling the single qubit rotations. \Cref{qc:Controlled Heisenberg and Ising} shows the according Heisenberg and Ising gates.
\begin{figure}[b]
  \center
  \includegraphics{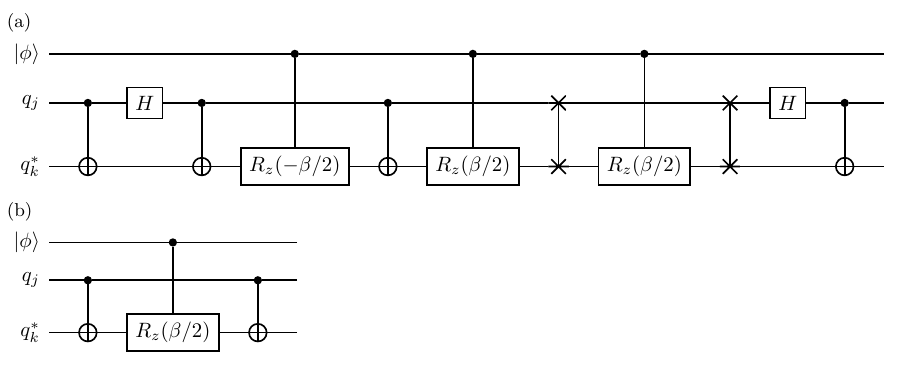}
  \captionShortcut[Controlled Heisenberg and Ising]{Controlled versions of the Heisenberg (a) and Ising (b) gates from \cref{qc:Implementation of the Heisenberg time evolution
  operator,qc:Implementation of the Ising time evolution operator} and \cref{eq:evolution helpers}, with $\beta = 4
\alpha$. The starred qubit indicates the central qubit (potentially after some swap operations) and $\ket{\phi}$ is the control qubit.}{qc}
\end{figure}
The structure of the whole circuit is similar to the one described in \cref{qc:First layout}, however the part containing the total angular momentum can be optimized to require fewer swap gates: Firstly, one should swap the roles of control and target qubit for the Ising-like terms $U\pa{I_q, t}$ as shown in \cref{qc:Controlled multi Ising}; secondly, the additional controlled rotation $\ee^{-i\frac{t}{\hbar} g \sigma^z_q / 2}$, from the last term in \cref{eq:evolution splitting}, should be executed directly after $U\pa{I_q, t}$.
\begin{figure}[H]
\center
\includegraphics{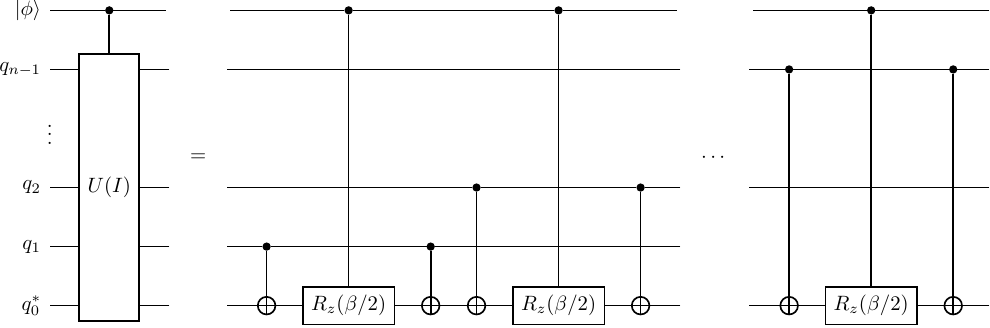}
\captionShortcut[Controlled multi Ising]{Controlled version of the the gate in \cref{qc:First layout}d. The starred qubit indicates the central qubit and $\ket{\phi}$ is the control qubit.}{qc}
\end{figure}

\twocolumngrid

\bibliographystyle{custom_IEEEtran}
\bibliography{literature}

\begin{acronym}[NISQ]
\acro{bcs}[BCS]{Bardeen-Cooper-Schrieffer \acroextra{superconductivity theory}}
\acro{css}[CSS]{central spin system}
\acro{nisq}[NISQ]{Noisy Intermediate-Scale Quantum}
\end{acronym}

\end{document}